\documentclass[a4paper]{article}

\usepackage{fullmacros}
\usepackage{subfig}
\usepackage{multirow}
\usepackage{colortbl}

\def\arrowedvec{\mathaccent"017E}
\newcommand{\arrowedvecstar}{\arrowedvec{\star}}

\renewcommand{\ext}[1]{#1^{\rightleftarrows}}

\newcommand{\measured}[1]{\mathbf{#1}}

\newcommand{\graphing}[1]{[{\mathcal #1}]}

\renewcommand{\measure}[1]{\measured{#1}}
\renewcommand{\model}[2]{\mathbb{M}[#1,\mathfrak{#2}]}

\newcommand{\icc}{{\sc icc}\xspace}
\newcommand{\gct}{{\sc gct}\xspace}
\newcommand{\pram}{{\sc pram}\xspace}
\newcommand{\prams}{{\sc pram}s\xspace}
\newcommand{\GoI}{{\sc goi}\xspace}
\newcommand{\IG}{{\sc ig}\xspace}

\newcommand{\acton}{\curvearrowright}

\newcommand{\altpath}[1]{\mathrm{AltPath}(#1)}
\newcommand{\seq}[1]{\mathbf{#1}}

\newcommand{\pushone}{{\tt push}_{1}}
\newcommand{\pushzero}{{\tt push}_{0}}
\newcommand{\pushstar}{{\tt push}_{\star}}
\newcommand{\pop}{{\tt pop}}

\renewcommand{\word}[1]{{\mathtt #1}}

\renewcommand{\twwprojects}{\cond{\sharp W}_{2}}
\renewcommand{\graphtw}[1]{W_{\word{#1}}}
\renewcommand{\graphingtw}[1]{W_{\word{#1}}}
\renewcommand{\graphtw}[1]{\bar{W}_{\word{#1}}}
\renewcommand{\ListType}{\cond{Nat}_{2}}
\renewcommand{\reptw}[1]{{\tt Gp}(\word{#1})}

\renewcommand{\graphing}[1]{[{\mathcal #1}]}
\renewcommand{\repany}[1]{\mathbf{Rep}(\word{#1})}

\newcommand{\halfidentity}[1][]{\mathrm{Id}^{1/2}_{#1}}

\newcommand{\pred}[2][]{\mathbf{Pred}^{\mathtt{#1}}(\microcosm{#2})}
\newcommand{\predcondet}[1]{\pred[co-ndet]{#1}}
\newcommand{\predndet}[1]{\pred[ndet]{#1}}
\newcommand{\preddet}[1]{\pred[det]{#1}}
\newcommand{\predprob}[1]{\pred[prob]{#1}}
\newcommand{\ListTypeELL}{{\rm BList}}

\newcommand{\testdet}{\testdetpos}

\AtBeginDocument{%
  \providecommand\BibTeX{{%
    \normalfont B\kern-0.5em{\scshape i\kern-0.25em b}\kern-0.8em\TeX}}}

\begin{document}

\title{Probabilistic Complexity Classes through Semantics}

\author{Thomas Seiller -- \texttt{seiller@lipn.fr}\\ %
  CNRS \& University Sorbonne Paris North, France
}



\maketitle

\begin{abstract}
In a recent paper, the author has shown how Interaction Graphs models for linear logic can be used to obtain implicit characterisations of non-deterministic complexity classes. In this paper, we show how this semantic approach to Implicit Complexity Theory (ICC) can be used to characterise deterministic and probabilistic models of computation. In doing so, we obtain correspondences between group actions and both deterministic and probabilistic hierarchies of complexity classes. As a particular case, we provide the first implicit characterisations of the classes \PLogspace (unbounded error probabilistic logarithmic space) and \PPtime (unbounded error probabilistic polynomial time)
\end{abstract}

\section{Introduction}

Complexity theory finds its root in three different papers that, in the span of a single year, tackled the difficult question of defining a notion of \emph{feasible} computation \cite{ cobham,edmonds65,hartmanisstearns}. Interestingly, the three authors independently came up with the same answer, namely the class \FPtime of polynomial time computable functions. Indeed, while computability finds its roots in mathematical logic, i.e. very far from actual computing devices, it became quite clear at the time that the notion of \emph{computable function} does not fall into any reasonable notion of \emph{effective computability}. 
Based on this, the fields of complexity theory, whose aim is the definition and classification of functions based on how much resources (e.g. time, space) are needed to compute them, quickly developed. 

While progress on the classification problem was quick in the early days, new results quickly became more and more scarcer. The difficulty of the classification problem can be explained in several ways. First, from a logical point of view, the question of showing whether a complexity class cannot compute a given function corresponds to showing the negation of an existential statement. But the severe difficulty of this problem can be understood through negative results known as \emph{barriers}, i.e. results stating that currently known methods cannot solve current open problems. While three such barriers exist, we will take the stand that only two \emph{conceptual} barriers exists, the \emph{algebrization} barrier being understood as a (far-reaching) refinement of the older \emph{relativization} barrier.

\paragraph{Barriers.} For the following short discussion about barriers, let us consider the famous \Ptime vs. \NPtime problem. The relativisation barrier is based on Baker, Gill and Solovay result \cite{relativisation} that there exists (recursive) oracles $\mathcal{A,B}$ such that $\Ptime^{\mathcal{A}}=\NPtime^{\mathcal{A}}$ and $\Ptime^{\mathcal{B}}\neq\NPtime^{\mathcal{B}}$. This implies that any proof method that is oblivious to the dis/use of oracles -- in other words, that \emph{relativises} -- will not answer the \Ptime vs. \NPtime problem. The \emph{algebrization} barrier \cite{algebrisation} provides a conceptually similar but refined negative result for algebraic methods, i.e. for those proof methods that are oblivious to the dis/use of \emph{algebraic extensions} of oracles -- which are said to \emph{algebrize}. The third barrier provides a somewhat orthogonal negative result, based on a notion of \emph{natural proofs}: roughly speaking a proof is natural if it can be formulated as separating two complexity classes using a predicate on boolean functions that satisfies three properties: \emph{largeness} -- i.e. most problems in the smallest of the two classes will satisfy the predicate --, and \emph{constructibility} -- i.e. the predicate is decidable in exponential time.
Razborov and Rudich showed \cite{naturalproofs} that under the assumption of the existence of exponentially-hard random generators, a natural proof cannot be used to prove that $\Ptime\neq\NPtime$.

These three barriers altogether capture all known separation methods to date, and therefore state that proving new separation results will \enquote{require radically new techniques} \cite{algebrisation}. In the last twenty years, one research program has been thought of as the only serious proposal for developing new separation methods that would circumvent the barriers, namely Mulmuley's Geometric Complexity Theory (\gct) programme \cite{GCT1,GCT2,GCTsurvey1,GCTsurvey2}. Inspired from a geometric proof of lower bounds on a (weakened) variation of the \pram model \cite{MulmuleyPRAM}, Mulmuley \gct methods are based on involved methods from algebraic geometry. However, after twenty years of existence, the \gct program has been dented with a few negative results in the last years.

\paragraph{Graphings.} The author recently proposed a new approach to complexity theory, one which may lead to new proof methods for separation \cite{seiller-towards,seiller-tacl,seiller-goinda}. Although this claim might not be formally justified at this point, let us point out that a recent preprint uses these methods to recast and improve lower bounds in algebraic complexity \cite{seiller-pramsLB}. The principal idea behind the approach is to propose a new general mathematical theory of computation that accounts for the dynamics of programs. In essence, the guiding intuition is that a computation should be mathematically modelled as a dynamical system, in the same way physical phenomena are. Obviously, while a computation (i.e. a run of a program) is deterministic and can be represented as such, a \emph{program} is not in general: it might be e.g. probabilistic, deterministic, and represent in itself several possible runs on a given input. Seiller's proposal is therefore to work with generalisations of dynamical systems introduced under the name of \emph{graphings}. Similarly to dynamical systems, graphings come in three different flavours: discrete, topological and measurable. The distinction between those does not impact the following discussion, although the results in this paper will use \emph{measurable graphings}.

Graphings were initially introduced in the context of ergodic theory \cite{adams,gaboriaul2}, and entered the realm of theoretical computer science through work on semantics of linear logic \cite{seiller-goig,seiller-goie,seiller-goif}. The main result in this aspect is that a collection of graphings built from a monoid action onto a space $\alpha: M\acton X$ -- i.e. graphings that can be locally identified with endomorphisms of the type $\alpha(m)$ -- gives rise to a model of (fragments of) linear logic. It is in this context that Seiller discovered that the action $\alpha: M\acton X$ can be put into correspondence with complexity classes, i.e. choosing different monoid actions leads to models in which the represented programs are of limited complexity, formalising an intuition that already appeared in the more involved context of operator algebras \cite{seiller-masas}. The first result in this direction \cite{seiller-goinda} provided a correspondence between a hierarchy of group actions $\microcosm{m}_{1},\microcosm{m}_{2},\dots,\microcosm{m}_{k},\dots,\microcosm{m}_{\infty}$ and a hierarchy of complexity classes between (and including) \Regular -- the classe of regular languages -- and \coNLogspace.

\paragraph{Contributions.} The current paper provides similar characterisations of the corresponding deterministic, non-deterministic (with the usual notion of acceptance), and probabilistic hierarchies. This is an important step in the overall program, as it shows the techniques apply to several computational paradigms. From a more external point of view, the techniques provides the first implicit characterisation probabilistic complexity classes, such as  \PLogspace (resp. \PPtime) of problems decidable (with unbounded error) by a probabilistic machine using logarithmic space (resp. polynomial time) in the input. Figure \ref{contribs} recapitulates\footnote{This table is not exhaustive, but shows the most common classes. In particular, this paper also characterises numerous classes not shown here, notably the classes of languages recognized by $k$-head two-way atutomata with a pushdown stack where $k$ is a fixed integer.} the known characterisations through the author's methods, showing the results of the current paper in white cells (previous results are shown in gray cells).

\begin{figure*}
\begin{center}
\begin{tabular}{|c||c|c|c|c|}
\hline
Microcosm & deterministic model & \multicolumn{2}{c|}{non-deterministic model} & probabilistic model \\\hline
$\microcosm{m}_{1}$ & \cellcolor{blue!10}$\Regular$ & \cellcolor{blue!10}\Regular & \Regular & \cellcolor{blue!10}\Stochastic  \\
$\vdots$&$\cellcolor{blue!10}\vdots$&$\cellcolor{blue!10}\vdots$&$\vdots$&$\cellcolor{blue!10}\vdots$\\
$\microcosm{m}_{k}$ & \cellcolor{blue!10}$\textsc{d}_{k}$ & \cellcolor{blue!10}$\textsc{n}_{k}$& $\textsc{co-n}_{k}$ & \cellcolor{blue!10}$\textsc{p}_{k}$ \\
$\vdots$&\cellcolor{blue!10}$\vdots$&\cellcolor{blue!10}$\vdots$&$\vdots$&\cellcolor{blue!10}$\vdots$\\
$\microcosm{m}_{\infty}$ & \cellcolor{blue!10}\Logspace & \cellcolor{blue!10}\NLogspace& \coNLogspace & \cellcolor{blue!10}\PLogspace  \\
\hline
$\microcosm{n}$ & \cellcolor{blue!10}\Ptime & \cellcolor{blue!10}\Ptime& \cellcolor{blue!10}\Ptime & \cellcolor{blue!10}\PPtime  \\
\hline
\end{tabular}
\end{center}
\caption{Known semantic characterisations of predicate complexity classes. Our contributions are shown in blue cells.\newline \small $\textsc{d}_{k}$ (resp. $\textsc{n}_{k}$, $\textsc{p}_{k}$) is the languages decided by two-way k-heads (resp. nondeterministic, probabilistic) automata.}\label{contribs}
\end{figure*}


It is important to understand that this approach can be understood both from a logical and a computability point of view. While coming from models of linear logic, the techniques relate to \emph{Implicit Computational Complexity} (\icc) and can be seen from this perspective as a semantic variant of constraint linear logics. However the modelling of programs as graphings go beyond the usual scope of the Curry-Howard correspondence, allowing for instance the sound representation of \pram machines\footnote{Although the cited work shows how to interpret some algebraic variant of \prams, it should be clear to the reader that the techniques used apply immediately to the usual notion of \prams.} \cite{seiller-pramsLB}. From this perspective, one could argue for this approach as a "computing with dynamical systems", with the underlying belief that any computation might be represented as such.

\section{Interaction Graphs and Complexity}

Interaction Graphs models of linear logic were developed in order to generalise Girard's geometry of interaction constructions to account for \emph{quantitative} aspects, in particular adapting to non-deterministic and probabilistic settings. The aim of the \GoI (and hence the \IG) approach is to obtain a \emph{dynamic} model of proofs and their cut-elimination procedure. 

\subsection{Graphings, Execution and Measurement}

In the general setting, graphings can be defined in three different flavours: discrete, topological and measurable. In this paper, we will be working with measurable graphings, and will refer to them simply as \emph{graphings}. Graphings act on a chosen space (hence here, on a measured space); the definition of graphings makes sense for any measured space $\measured{X}$, and under some mild assumptions on $\measured{X}$ it provides a model of (at least) Multiplicative-Additive Linear Logic (\mall) \cite{seiller-goig,seiller-goie,seiller-goif}. We now fix the measure space of interest in this paper.

\begin{definition}[The Space]
We define the measure space $\measure{X}=\realN\times[0,1]^{\naturalN}\times\{\star,0,1\}^{\naturalN}$ where $\realN\times[0,1]^{\naturalN}$ is considered with its usual Borel $\sigma$-algebra and Lebesgue measure. The space $\{\star,0,1\}^{\naturalN}$ is endowed with the natural topology\footnote{I.e. the topology induced by basic \emph{cylindrical} open sets $V(w)=\{ f:\naturalN\rightarrow\{0,1\} \mid \forall i\in [\lg{w}], f(i)=w_i\}$ where $w$ is a finite word on $\{0,1\}$.}, the corresponding Borel $\sigma$-algebra and the natural measure given by $\mu(V(w))=3^{-\lg{w}}$.
\end{definition}

Borrowing the notation introduced in previous work \cite{seiller-goif,seiller-goinda}, we denote by $(x,\seq{s},\seq{\pi})$ the points in $\measure{X}$, where $\seq{s}$ and $\seq{\pi}$ are sequences for which we allow a concatenation-based notation, e.g.\ we write $(a,b)\cdot \seq{s}$ for the sequences whose first two elements are $a,b$ (and we abusively write $a\cdot\seq{s}$ instead of $(a)\cdot\seq{s}$). Given a permutation $\sigma$ over the natural numbers, we write $\sigma(\seq{s})$ the result of its natural action on the $\naturalN$-indexed list $\seq{s}$.

Graphings are then defined as objects acting on the measured space $\measured{X}$. A parameter in the construction allows one to consider subsets of graphings based on \emph{how} they act on the space. To do this, we fix a monoid of measurable maps that we call a \emph{microcosm}. Again, while graphings can be defined in full generality, some conditions on the chosen microcosms are needed to construct models of \mall. The following microcosms, which are of of interest in this paper, do satisfy these additional requirements.

\begin{definition}[Microcosms]
For all integer $i\geqslant 1$, we consider the translations 
\[ \ttrm{t}_{z}:(x,\seq{s},\seq{\pi})\mapsto(x+z,\seq{s},\seq{\pi}) \] 
for all integer $z$, and the permutations 
\[ \ttrm{p}_{\sigma}:(x,\seq{s},\seq{\pi})\mapsto(x,\sigma(\seq{s}),\seq{\pi}) \]
for all bijection $\sigma:\naturalN\rightarrow\naturalN$ such that $\sigma(k)=k$ for all $k>i$. 
We denote by $\microcosm{m}_i$ the monoid generated by those maps, and by $\microcosm{m}_{\infty}$ the union $\cup_{i>1}\microcosm{m}_{i}$.

We also consider the maps:
\[ \pop: (x,\seq{s},c\cdot\seq{pi})\mapsto (x,\seq{s},\seq{pi}) \]
\[ \pushzero: (x,\seq{s},\seq{pi})\mapsto (x,\seq{s},0\cdot\seq{pi}) \]
\[ \pushone: (x,\seq{s},\seq{pi})\mapsto (x,\seq{s},1\cdot\seq{pi}) \]
\[ \pushstar: (x,\seq{s},\seq{pi})\mapsto (x,\seq{s},\star\cdot\seq{pi}) \]
We denote by $\microcosm{n}_i$ the monoid generated the microcosm $\microcosm{m}_i$ extended by those three maps, and by $\microcosm{n}_{\infty}$ the union $\cup_{i>1}\microcosm{n}_{i}$.

Finally, let us denote by $a\bar{+}b$ the fractional part of the sum $a+b$.
We also define the microcosms $\bar{\microcosm{m}}_{i}$ (resp. $\bar{\microcosm{n}}_{i}$) as the smallest microcosms containing $\microcosm{m}_{i}$ (resp. $\microcosm{n}_{i}$) and all translations $\ttrm{t}_{\lambda}:(x,a\cdot \seq{s})\mapsto(x, (a\bar{+}\lambda)\cdot\seq{s})$ for $\lambda$ in $[0,1]$.
\end{definition}

We are now able to define graphings. Those are formally defined as quotient of graph-like objects called \emph{graphing representatives}\footnote{In earlier works \cite{seiller-goig}, the author did not introduce separate terminologies for graphings and graphing representatives. While this can be allowed since all operations considered on graphings are compatible with the quotient, we chose here a more pedagogical approach of introducing a clear distinction between those types of objects.}. Graphing representatives are just (countable) families of weighted \emph{edges} defined by a source -- a measurable subset of $\measured{X}\times D$ where $S$ is a fixed set of \emph{states} -- and  a \emph{realiser}, i.e. a map in the considered microcosm and a new state in $D$. 

As in previous work \cite{seiller-goinda}, we fix the monoid of weights of graphings $\Omega$ to be equal to $[0,1]\times\{0,1\}$ with usual multiplication on the unit interval and the product on $\{0,1\}$. To simplify notations, we write elements of the form $(a,0)$ as $a$ and elements of the form $(a,1)$ as $a\cdot \mathbf{1}$. On this set of weights, we will consider the fixed parameter map $m(x,y)=xy$ in the following (used in \autoref{def:measurement}.

\begin{definition}[Graphing representative]
We fix a measure space $\measure{X}$, a microcosm $\microcosm{m}$, a monoid $\Omega$, a measurable subset $V^{G}$ of $\measure{X}$, and a finite set $S^{G}$. A ($\Omega$-weighted) \emph{$\microcosm{m}$-graphing representative} $G$ of \emph{support} $V^{G}$ and \emph{stateset} $S^{G}$ is a countable set 
$$\{(S^{G}_{e},\ttrm{i}^{G}_{e},\ttrm{o}^{G}_{e},\phi^{G}_{e},\omega^{G}_{e})~|~e\in E^{G}\},$$
where $S^{G}_{e}$ is a measurable subset\footnote{As $D^{G}$ is considered as a discrete measure space, a measurable subset of the product is simply a finite collection of measurable subset indexed by elements of $D^{G}$.} of $V^{G}\times D^{G}$, $\phi^{G}_{e}$ is an element of $\microcosm{m}$ such that $\phi_{e}^{G}(S^{G}_{e})\subseteq V^{G}$, $\ttrm{i}^{G}_{e},\ttrm{o}^{G}_{e}$ are elements of $D^{G}$, and $\omega^{G}_{e}\in\Omega$ is a \emph{weight}. We will refer to the indexing set $E^{G}$ as the set of \emph{edges}. For each edge $e\in E^{G}$ the set $S^{G}_{e}\times\{\ttrm{i}^{G}_{e}\}$ is called the source of $e$, and we define the \emph{target} of $e$ as the set $T^{G}_{e}\times\{\ttrm{o}^{G}_{e}\}$ where $T^{G}_{e}=\phi^{G}_{e}(S^{G}_{e})$.
\end{definition}

The notion of graphing representatives captures more than just an action on a space. Indeed, consider a graphing representative $G$ with a single edge of source $S\disjun S'$, weight $w$ and realiser $f$. Then the graphing representative $H$ having two edges of weight $w$, realiser $f$ and respective sources $S$ and $S'$ intuitively represent the same action on $\measured{X}$ as $G$. In other words, the notion of graphing representative captures both the notion of action and a notion of representation of this action. The auuthor therefore defines a notion of \emph{refinement} which allows for the consideration of the equivalence $G\sim H$ that captures the fact that $G$ and $H$ represent the same action. 

%

\begin{definition}[Refinements]
Let $F,G$ be graphing representatives. Then $F$ is a refinement of $G$ -- written $F\leqslant G$ -- if there exists a partition\footnote{We allow the sets $E^{F}_{e}$ to be empty.} $(E^{F}_{e})_{e\in E^{G}}$ of $E^{F}$ s.t. $\forall e\in E^{G}, \forall f,f'\in E^{F}_{e}$:
\begin{itemize}[noitemsep,nolistsep]
\item $\omega^{F}_{f}=\omega_{e}^{G}\textrm{ and }\phi^{F}_{f}=\phi_{e}^{G}$;
\item $\cup_{f\in E^{F}_{e}} S^{F}_{f} =_{a.e.} S^{G}_{e}\textrm{ and }f\neq f' \Rightarrow \mu(S^{F}_{f} \cap  S^{F}_{f'})=0$.
\end{itemize}
\end{definition}

Then two graphing representatives are \emph{equivalent} if and only if they possess a common refinement. The actual notion of \emph{graphing} is then an equivalence class of the objects just defined w.r.t. this equivalence. 

\begin{definition}[Graphing]
A \emph{graphing} is an equivalence class of graphing representatives w.r.t. the induced  equivalence:
$$F\sim G\Leftrightarrow \exists H,~H\leqslant F\textrm{ and }H\leqslant G$$
\end{definition}

Since all operations considered on graphings were shown to be compatible with this quotienting \cite{seiller-goig}, i.e.\ well defined on the equivalence classes, we will in general make no distinction between a graphing -- as an equivalence class -- and a graphing representative belonging to this equivalence class.

\begin{definition}[Execution]\label{def:execution}
Let $F$ and $G$ be graphings such that $V^{F}=V\disjun C$ and $V^{G}=C\disjun W$ with $V\cap W$ of null measure. Their \emph{execution} $F\plug G$ is the graphing of support $V\disjun W$ and stateset $D^{F}\times D^{G}$ defined as the set of all\footnote{We refer to the author's work on graphing for the full definition. Intuitively, if $\pi$ is a path, then $[\pi]_{o}^{o}(C)$ denotes the restriction $\pi$ to the points in the source that lie outside the set $C$ and are mapped by the realiser of the path to an element outside of $C$.} $[\pi]_{o}^{o}(C)$ where $\pi$ is an element of $\altpath{F,G}$, the set of alternating path between $F$ and $G$, i.e. the set of paths $\pi=e_1e_2\dots e_n$ such that for all $i=1,\dots,n-1$, $e_i\in F$ iff $e_{i+1}\in G$.
$$ 
\begin{array}{rcl} F\plug G&=&\left\{\left(S^{\decoupe C}_{\pi},(\ttrm{i}_{e_{1}},\ttrm{i}_{e_{2}}),(\ttrm{o}_{e_{n-1}},\ttrm{o}_{e_{n}}),\phi_{\pi},\omega_{\pi}\right)\right.\\
&&\hspace{2em}\left.|~\pi= e_{1},e_{2},\dots,e_{n}\in\altpath{F,G}\right\}
\end{array} 
$$
\end{definition}

\subsection{Measurement, Proofs and Types}

We now recall the notion of measurement. Since in the specific case that will be of interest to us, i.e. when restricting to the microcosms considered in this paper, the expression of the measurement can be simplified, we only give this simple expression and point the curious reader to earlier work for the general definition \cite{seiller-goig}.

\begin{definition}\label{def:measurement}
The measurement between two graphings (realised by measure-preserving maps) is defined as 
$$\meas[]{F,G}=\sum_{\pi\in\repcirc{F,G}} \int_{\supp{\pi}} \frac{m(\omega(\pi)^{\rho_{\phi_{\pi}}(x))}}{\rho_{\phi_{\pi}}(x)}d\lambda(x)$$
where $\rho_{\phi_{\pi}(x)}=\inf\{n\in\naturalN~|~ \phi_{\pi}^{n}(x)=x\}$ (here $\inf\emptyset=\infty$), and the support 
$\supp{\pi}$ of $\pi$ is the set of points $x$ belonging to a finite orbit \cite[Definition 41]{seiller-goig}.\end{definition}

The measurement is used to define linear negation. But first, let us recall the notion of \emph{project} which is the semantic equivalent of \emph{proofs}.

\begin{definition}
A project of support $V$ is a pair $(a,A)$ of a real number $a$ and a finite formal sum $A=\sum_{i\in I} \alpha_{i}A_{i}$ where for all $i\in I$, $\alpha_{i}\in\realN$ and $A_{i}$ is a graphing of support $V$.
\end{definition}

We can then define an \emph{orthogonality relation} on the set of projects. Orthogonality captures the notion of linear negation and somehow translates the correctness criterion for proof nets. Its definition is based on the measurement defined above, extended to formal weighted sums of graphings by \enquote{linearity} \cite{seiller-goiadd,seiller-goig}.

\begin{definition}
Two projects $(a,A)$ and $(b,B)$ are orthogonal -- written $(a,A)\poll{}(b,B)$ -- when they have equal support and $\meas[]{(a,A),(b,B)}\neq 0,\infty$. We also define the orthogonal of a set $E$ as $E^{\pol}=\{(b,B): \forall (a,A)\in A, (a,A)\poll{}(b,B)\}$ and write $E^{\pol\pol}$ the double-orthogonal $(E^{\pol})^{\pol}$.
\end{definition}

Orthogonality allows for a definition of types. In fact the models are defined based on two notions of types --\emph{conducts} and \emph{behaviours} \cite{seiller-goiadd}. Conducts are simple to define but while their definition is enough to define a model of multiplicative linear logic, dealing with additives requires the more refined notion of \emph{behaviour}. 

\begin{definition}
A \emph{conduct} of support $V^{A}$ is a set $\cond{A}$ of projects of support $V^{A}$ such that $\cond{A}=\cond{A}^{\pol\pol}$. A \emph{behaviour} is a conduct such that for all $(a,A)$ in $\cond{A}$ (resp. $\cond{A}^{\pol}$) and for all $\lambda\in\realN$, $(a,A+\lambda\emptyset)$ belongs to $\cond{A}$ (resp. $\cond{A}^{\pol}$) as well. If both $\cond{A}$ and $\cond{A}^{\pol}$ are non-empty, we say $\cond{A}$ is \emph{proper}.
\end{definition}

Conducts provide a model of Multiplicative Linear Logic. The connectives $\otimes,\multimap$ are defined as follows: if $\cond{A}$ and $\cond{B}$ are conducts of disjoint supports $V^{A}, V^{B}$, i.e. $V^{A}\cap V^{B}$ is of null measure, then:
\[
\begin{array}{rcl}
\cond{A\otimes B}&=&\{\de{a\plug b}~|~\de{a}\in\cond{A},\de{b}\in\cond{B}\}^{\pol\pol}\\
\cond{A\multimap B}&=&\{\de{f}~|~\forall \de{a}\in\cond{A},\de{f\plug b}\in\cond{B}\}.
\end{array}
\]
However, to define additive connectives, one has to restrict the model to behaviours. In this paper, we will deal almost exclusively with proper behaviours. Based on the following proposition, we will therefore consider mostly projects of the form $(0,L)$ which we abusively identify with the underlying sliced graphing $L$. Moreover, we will use the term \enquote{behaviour} in place of \enquote{proper behaviour}.

\begin{proposition}[{\cite[Proposition 60]{seiller-goiadd}}]
If $\cond{A}$ is a proper behaviour, $(a,A)\in\cond{A}$ implies $a=0$.
\end{proposition}

Finally, let us mention the fundamental theorem for the interaction graphs construction in the restricted case we just exposed\footnote{The general construction allows for other sets of weights as well as whole families of measurements \cite{seiller-goig}.}.

\begin{theorem}[{\cite[Theorem 1]{seiller-goig}}]\label{thm:mallmodels}
For any microcosm $\microcosm{m}$, the set of behaviours provides a model of Multiplicative-Additive Linear Logic (\MALL) without multiplicative units.
\end{theorem}

\subsection{Integers, Machines and Complexity}

We now recall some definitions introduced in previous work by the author characterising complexity classes by use of graphings \cite{seiller-goinda}; interested readers will find there (and in some references therein \cite{seiller-conl,seiller-lsp}) more detailed explanations of -- and motivations for -- the definitions. In particular the representation of binary words is related to 
\[ \ListTypeELL := \forall X~ \oc(X\multimap X)\multimap \oc(X\multimap X)\multimap \oc(X\multimap X), \] 
the type of binary lists in Elementary Linear Logic \cite{LLL,danosjoinet}. 

\begin{notation}
To ease notations, we only consider words over $\Sigma=\{0,1\}$ in this paper. We write $\ext{\Sigma}$ the set $\{0,1,\star\}\times\{\In,\Out\}$. We also denote by $\vertices{\Sigma}$ the set $\ext{\Sigma}\cup\{\ttrm{a},\ttrm{r}\}$, where $\ttrm{a}$ (resp. $\ttrm{r}$) stand for $\accept$ (resp. $\reject$). 

Initial segments of the natural numbers $\{0,1,\dots,n\}$ are denoted $[n]$. Up to renaming, all statesets can be considered to be of this form.
\end{notation}

\begin{notation}
As in the previous paper \cite{seiller-goinda}, we fix once and for all an injection $\Psi$ from the set $\vertices{\Sigma}$ to intervals in $\realN$ of the form $[k,k+1]$ with $k$ an integer. For all $f\in\vertices{\Sigma}$ and $Y$ a measurable subset of $[0,1]^{\naturalN}$, we denote by $\bracket{f}_{Y}^{Z}$ the measurable subset $\Psi(f)\times Y\times Z$ of $\measure{X}$, where $Y\subset[0,1]^{\naturalN}$ and $Z\subset\{\star,0,1\}^{\naturalN}$. If $Y=[0,1]^{\naturalN}$ (resp. $Z=\{\star,0,1\}^{\naturalN}$), we omit the subscript (resp. superscript). The notation is extended to subsets $S\subset\vertices{\Sigma}$ by $\support{S}=\cup_{f\in S}\support{f}$ (a disjoint union).
\end{notation}

Given a word $\word{w}=\star a_{1}a_{2}\dots a_{k}$, we denote $\graphtw{w}$ the graph with set of vertices $V^{\graphtw{w}}\times D^{\graphtw{w}}$, set of edges $E^{\graphtw{w}}$, source map $s^{\graphtw{w}}$ and target map $t^{\graphtw{w}}$ respectively defined as follows:
\begin{equation*}
\begin{array}{l}
V^{\graphtw{w}}=\ext{\Sigma} \hspace{2em} D^{\graphtw{w}}=[k] \hspace{2em}  E^{\graphtw{w}}=\{r,l\}\times[k]\\
\begin{array}{rcl}
s^{\graphtw{w}}&=&(r,i)\mapsto (a_{i},\Out,i)\\
		&&(l,i)\mapsto (a_{i}, \In,i)\\
t^{\graphtw{w}}&=&(r,i)\mapsto (a_{i+1},\In,i+1 \textnormal{ mod } k+1)\\
		&&(l,i)\mapsto (a_{i-1}, \Out,i-1 \textnormal{ mod } k+1)
\end{array}
\end{array}
\end{equation*}
This graph is the discrete representation of $\word{w}$. Detailed explanations on how these graphs relate to the proofs of the sequent $\vdash \ListTypeELL$ can be found in earlier work \cite{seiller-phd,seiller-conl}.

\begin{definition}
Let $\word{w}$ be a word $\word{w}=\star a_{1}a_{2}\dots a_{k}$ over the alphabet $\Sigma$. We define the \emph{word graphing} $\graphingtw{w}$ of support $\wordsupport$ and stateset $D^{\graphtw{w}}$ by the set of edges $E^{\graphtw{w}}$ and for all edge $e$:
\[\{(\bracket{f},i,j,\phi_{f,i}^{g,j},1)\mid e\in E^{\graphtw{w}}\}, \]
where $s^{\graphtw{w}}(e)=(f,i)$, $t^{\graphtw{w}}(e)=(g,j)$, and $\phi_{f,i}^{g,j}:(\bracket{f},x,i)\mapsto (\bracket{g},x,j)\}$.
\end{definition}

\begin{notation}
We write $\reptw{w}$ the set of word graphings for $\word{w}$. It is defined as the 
set of graphings obtained by renaming the stateset $D^{\graphtw{w}}$ w.r.t. an injection $[k]\rightarrow [n]$.
\end{notation}

\begin{definition}
Given a word $\word{w}$, a \emph{representation of $\word{w}$} is a graphing $\oc L$ where $L$ belongs to $\reptw{w}$. The set of representations of words in $\Sigma$ is denoted $\twwprojects$, the set of representations of a specific word $\word{w}$ is denoted $\repany{w}$. 

We then define the conduct $\oc\ListType=\cond{(\twwprojects)^{\pol{}\pol{}}}$.
\end{definition}

\begin{definition}
We define the (unproper) behaviour $\NBool$ as $\cond{T}_{\resultsupport}$, where for all measurable sets $V$ the behaviour $\cond{T}_{V}$ is defined as the set of all projects of support $V$.
For all microcosms $\microcosm{m}$, we define $\pred{m}$ as the set of $\microcosm{m}$-graphings in $\oc\ListType\multimap\NBool$.
\end{definition}

\begin{definition}
An $\microcosm{m}$-graphing $G$ is \emph{finite} when it has a representative $H$ whose set of edges $E^{H}$ is finite.
\end{definition}

\begin{definition}
A \emph{nondeterministic predicate $\microcosm{m}$-machine} over the alphabet $\Sigma$ is a finite $\microcosm{m}$-graphing belonging to $\pred{m}$ with all weights equal to $1$.
\end{definition}

The computation of a given machine given an argument is represented by the \emph{execution}, i.e. the computation of paths defined in \autoref{def:execution}. The result of the execution is an element of $\NBool$, i.e.\ somehow a generalised boolean value\footnote{If one were working with \enquote{deterministic machines} \cite{seiller-towards}, it would belong to the subtype $\Bool$ of booleans.}.

\begin{definition}[Computation]
Let $M$ be a $\microcosm{m}$-machine, $\word{w}$ a word over the alphabet $\Sigma$ and $\oc L\in\oc\ListType$. The \emph{computation} of $\de{M}$ over $\oc L$ is defined as the graphing $M\plug \oc L\in\NBool$.
\end{definition}

\begin{definition}[Tests]\label{def:tests}
A \emph{test} is a family $\testfont{T}=\{\de{t}_{i}=(t_{i},T_{i})~|~i\in I\}$ of projects of support $\resultsupport$.
\end{definition}

We now want to define the language characterised by a machine. For this, one could consider \emph{existential} $\mathcal{L}_{\exists}^{\testfont{T}}(M)$ and  \emph{universal} $\mathcal{L}_{\forall}^{\testfont{T}}(M)$ languages for a machine $M$ w.r.t. a test $\testfont{T}$:
$$
\begin{array}{rcl}
\mathcal{L}_{\exists}^{\testfont{T}}(M)&=&\{\word{w}\in\Sigma^{\ast}~|~\forall \de{t}_{i}\in\testfont{T}, \exists \de{w}\in\repany{w}, M\plug \de{w}\poll{} \de{t}_{i}\}\\
\mathcal{L}_{\forall}^{\testfont{T}}(M)&=&\{\word{w}\in\Sigma^{\ast}~|~\forall \de{t}_{i}\in\testfont{T}, \forall \de{w}\in\repany{w}, M\plug \de{w}\poll{} \de{t}_{i}\}
\end{array}
$$

We now introduce the notion of uniformity, which describes a situation where both definitions above coincide. This collapse of definitions is of particular interest because it ensures that both of the following problems are easy to solve:
\begin{itemize}[nolistsep,noitemsep]
\item whether a word belongs to the language: from the existential definition one only needs to consider one representation of the word;
\item whether a word does not belong to the language: from the universal definition, one needs to consider only one representation of the word.
\end{itemize}

\begin{definition}[Uniformity]
Let $\microcosm{m}$ be a microcosm. The test $\testfont{T}$ is said \emph{uniform} w.r.t. $\microcosm{m}$-machines if for all such machine $M$, and any two elements $\de{w},\de{w'}$ in $\repany{w}$:
$$M\plug \de{w}\in \testfont{T}^{\pol} \text{ if and only if }M\plug \de{w}'\in \testfont{T}^{\pol}$$
Given a $\microcosm{m}$-machine $M$, we write in this case $\mathcal{L}^{\testfont{T}}(M)=\mathcal{L}_{\exists}^{\testfont{T}}(M)=\mathcal{L}_{\forall}^{\testfont{T}}(M)$.
\end{definition}

\begin{definition}
For $U\subset\mathbf{X}$, we define $\identity[U]$ as the graphing with a single edge and stateset $[0]$: $\{(\support{\ttrm{r}},0,0,x\mapsto x,1\cdot\mathbf{1})\}$.
 \end{definition}

\begin{definition}
We define the test $\testdetneg$ as the set consisting of the projects $\de{t}^{-}_{\zeta}=(\zeta,\identity[\support{\ttrm{r}}])$, where $\zeta\neq0$.
\end{definition}

\begin{proposition}
The test $\testdetneg$ is uniform w.r.t. $\microcosm{m}_{\infty}$-machines.
\end{proposition}

\begin{definition}
We define the complexity class $\predcondet{m}$ as the set
\[ \{\mathcal{L}^{\testdetneg}(M)\mid M\text{ $\microcosm{m}$-machine} \}. \]
\end{definition}

We recall the main theorem of the author's previous paper \cite{seiller-goinda}, and refer to \autoref{def:cc} for the definition of the characterised complexity classes.
\begin{theorem}
For all  $i\in\naturalN^{\ast}\cup\{\infty\}$, the class $\predcondet{{m}_{i}}$ is equal to \cctwconfa{i}. As particular cases, $\predcondet{{m}_{1}}=\Regular$ and $\predcondet{{m}_{\infty}}=\coNLogspace$.
\end{theorem}

\subsection{Characterising $\NLogspace$}

The starting point of this work was the realisation that one can define another test $\testdet$ that allows to capture the notion of acceptance in $\NLogspace$. Based on this idea, and using technical lemmas from the previous paper, we can characterise easily the hierarchy of complexity classes defined by $k$-head non-deterministic automata with the standard non-deterministic acceptance condition (i.e. there is at least one accepting run).

\begin{definition}
For $U\subset\mathbf{X}$, we define $\halfidentity[U]$ as the graphing with a single edge and stateset $[0]$: $\{(\support{\ttrm{r}},0,0,x\mapsto x,\frac{1}{2}\cdot\mathbf{1})\}$.
 \end{definition}
 
\begin{definition}
The test $\testdet$ is defined as the family:
\[ \{ (0,\halfidentity[\support{\ttrm{a}}_{[0,\frac{1}{n}]^{n}\times[0,1]^{\naturalN}}])\mid n\in\naturalN \}. \]
\end{definition}

\begin{definition}
We define the complexity class $\predndet{m}$ as the set
\[ \{\mathcal{L}^{\testdet}(M)\mid M\text{ $\microcosm{m}$-machine in $\nmodel{\{1\}}{m}$} \}. \]
\end{definition}

This test does indeed provide the right characterisation. In fact, the sole element $(0,\halfidentity[\support{\ttrm{a}}])$ is enough to obtain soundness, since the result relies on a result from the author's previous paper \cite[Proposition 46]{seiller-goinda} (which is not stated here, as it is generalised by \autoref{prop:generalised46} below). Thus, a $k$-heads two-way automaton $\automaton{M}$ accepts a word $\word{w}$ if and only if there exists at least one alternating path between the graphing translation $\autograph{M}$ of $M$ and the word representation $\oc \graphingtw{w}$ whose source and target is $\support{\ttrm{a}}_{Y}$. Thus, $\automaton{M}$ accepts $\word{w}$ if and only if there are alternating cycles between $\autograph{M}\plug \oc \graphingtw{w}$ and $\halfidentity[\support{\ttrm{a}}]$, i.e. if and only if $\meas[]{\autograph{M}\plug \oc \graphingtw{w},\halfidentity[\support{\ttrm{a}}]}\neq 0,\infty$.

However, the whole family of tests is needed to obtain completeness. Indeed, in the general case, it might be possible that a $\microcosm{m}_i$-machine $G$ passes the test $\{(0,\halfidentity[\support{\ttrm{a}}])\}$ by taking several paths through the execution $G\plug \oc \graphingtw{w}$ (hence creating a cycle of arbitrary length). In that  case, it is not clear that the existence of such a cycle can be decided with some automaton $M$. However, if $G\plug \oc \graphingtw{w}$ passes all tests in $\testdet$, it imposes the existence of a cycle of length 2 between $G\plug \oc \graphingtw{w}$ and $\halfidentity[\support{\ttrm{a}}]$, something that can be decided by an automaton. 

A simple adaptation of the arguments then provides a proof of the following. We omit the details for the moment, as the next sections will expose a generalisation of the technique that also applies to the probabilistic case and to automata with a pushdown stack. The definitions of the complexity classes involved in the statement are given in \autoref{def:cc}.

\begin{theorem}
For all $i\in\naturalN^{\ast}\cup\{\infty\}$, 
\[\predndet{m_{i}}=\text{\cctwnfa{i}}\hspace{2em} \predndet{n_{i}}=\text{\cctwnfastack{i}}.\]
In particular, $\predndet{m_{\infty}}=\Logspace$ and $\predndet{n_{\infty}}=4\Ptime$\footnote{Here we characterise $\Ptime$ and not $\NPtime$ as one may expect because non-determinism for pushdown machines do not add expressivity, as shown by Cook \cite{cookP} using memoization.} 
\end{theorem}

\section{Deterministic and Probabilistic Models}

\begin{definition}
A graphing $\graphing{G}=\{S^{G}_{e},\phi^{G}_{e},\omega^{G}_{e}~|~e\in E^{G}\}$ is \emph{deterministic} if all edges have weight equal to $1$ and the following holds:
\[ \mu\left(\left\{x\in \measure{X}~|~ \exists e,f\in E^{G}, e\neq f\text{ and }x\in S_{e}^{G}\cap S_{f}^{G}\right\}\right)=0 \]
\end{definition}

\begin{remark}
The notion is quite natural, and corresponds in the case of graphs to the. requirement the out-degree of all vertices to be less or equal to 1. 
\end{remark}

We now prove that the set of deterministic graphings is closed under composition, i.e. if $F,G$ are deterministic graphings, then their execution $F\plug G$ is again a deterministic graphing. This shows that the sets of deterministic and non-deterministic graphings define submodels of $\model{\Omega}{m}$.

\begin{lemma}
The set of deterministic graphings is closed under execution.
\end{lemma}

\begin{proof}
A \emph{deterministic graphing} $F$ satisfies that for every edges $e,f\in E^{F}$, $S^{F}_{e}\cap S^{F}_{f}$ is of null measure. Suppose that the graphing $F\plug G$ is not deterministic. Then there exists a Borel $B$ of non-zero measure and two edges $e,f\in E^{F\plug G}$ such that $B\subset S^{F\plug G}_{e}\cap S^{F\plug G}_{f}$. The edges $e,f$ correspond to paths $\pi_{e}$ and $\pi_{f}$ alternating between $F$ and $G$. It is clear that the first step of these paths belong to the same graphing, say $F$ without loss of generality, because the Borel set $B$ did not belong to the \emph{cut}. Thus $\pi_{e}$ and $\pi_{f}$ can be written $\pi_{e}=f_{0}\pi^{1}_{e}$ and $\pi_{f}=f_{0}\pi^{1}_{f}$. Thus the domains of the paths $\pi^{1}_{e}$ and $\pi_{f}^{1}$ coincide on the Borel set $\phi_{f_{0}}^{F}(B)$ which is of non-zero measure since all maps considered are non-singular. One can then continue the reasoning up to the end of one of the paths and show that they are equal up to this point. Now, if one of the paths ends before the other we have a contradiction because it would mean the the Borel set under consideration would be at the same time inside and outside the cut, which is not possible. So both paths have the same length and are therefore equal. Which shows that $F\plug G$ is deterministic since we have shown that if the domain of two paths alternating between $F$ and $G$ coincide on a non-zero measure Borel set, the two paths are equal (hence they correspond to the same edge in $F\plug G$).
\end{proof}

One can then check that the interpretations of proofs by graphings in earlier papers \cite{seiller-goig,seiller-goie,seiller-goif} are all deterministic. This gives us the following theorem as a corollary of the previous lemma.

\begin{theorem}[Deterministic model]
Let $\Omega$ be a monoid and $\microcosm{m}$ a microcosm. The set of $\Omega$-weighted \emph{deterministic} graphings in $\microcosm{m}$ yields a model, denoted by $\dmodel{\Omega}{m}$, of multiplicative-additive linear logic.
\end{theorem}

\begin{definition}
We define the complexity class $\preddet{m}$ as the set
\[ \{\mathcal{L}^{\testdet}(M)\mid M\text{ $\microcosm{m}$-machine in $\dmodel{\{1\}}{m}$} \}. \]
\end{definition}

\subsection{The Probabilistic Model}

One can also consider several other classes of graphings. We explain here the simplest non-classical model one could consider, namely that of \emph{sub-probabilistic graphings}. In order for this notion to be well-defined, one should suppose that the unit interval $[0,1]$ endowed with multiplication is a submonoid of $\Omega$.

\begin{definition}
A graphing $\graphing{G}=\{S^{G}_{e},\phi^{G}_{e},\omega^{G}_{e}~|~e\in E^{G}\}$ is \emph{sub-probabilistic} if all the edges have weight in $[0,1]$ and the following holds:
\[ \mu\left(\left\{x\in \measure{X}~|~ \sum_{e\in E^{G}, x\in S^{G}_{e}}\omega^{G}_{e}>1\right\}\right)=0 \]
\end{definition}
%

It turns out that this notion of graphing also behaves well under composition, i.e. there exists a \emph{sub-probabilistic} submodel of $\model{\Omega}{m}$, namely the model of \emph{sub-probabilistic graphings}.

\begin{theorem}
The set of sub-probabilistic graphings is closed under execution.
\end{theorem}

\begin{proof}
If the weights of edges in $F$ and $G$ are elements of $[0,1]$, then it is clear that the weights of edges in $F\plug G$ are also elements of $[0,1]$. We therefore only need to check that the second condition is preserved.

Let us denote by $\outset{F\plug G}$ the set of $x\in X$ which are source of paths whose added weight is greater than $1$, and by $\outset{F\cup G}$ the set of $x$ which are source of edges (either in $F$ or $G$) whose added weight is greater than $1$. First, we notice that if $x\in\outset{F\plug G}$ then either $x\in\outset{F\cup G}$, or $x$ is mapped -- through at least one edge -- to an element $y$ which is itself in $\outset{F\cup G}$. To prove this statement, let us write $\outpaths{x}$ (resp. $\outedges{x}$) the set of paths in $F\plug G$ (resp. edges in $F$ or $G$) whose source contain $x$. We know the sum of all the weights of these paths is greater than $1$, i.e. $\sum_{\pi\in\outpaths{x}}\omega(\pi)>1$. But this sum can be rearranged by ordering paths depending on theirs initial edge, i.e. $\sum_{\pi\in\outpaths{x}}\omega(\pi)=\sum_{e\in\outedges{x}}\sum_{\pi=e\rho\in\outpaths{x}^{e}}\omega(\pi)$, where $\outpaths{x}^{e}$ denotes the paths whose first edge is $e$. Now, since the weight of $e$ appears in all $\omega(e\rho)=\omega(e)\omega(\rho)$, we can factorize and obtain the following inequality.
\[ \sum_{e\in\outedges{x}}\omega(e)\left(\sum_{\pi=e\rho\in\outpaths{x}^{e}}\omega(\rho)\right)>1 \] 
Since the sum $\sum_{e\in\outedges{x}}\omega(e)$ is not greater than $1$, we deduce that there exists at least one $e\in\outedges{x}$ such that $\sum_{\pi=e\rho\in\outpaths{x}^{e}}\omega(\rho)>1$. However, this means that $\phi_{e}(x)$ is an element of $\outset{F\plug G}$.

Now, we must note that $x$ is not element of a cycle. This is clear from the fact that $x$ lies in the carrier of $F\plug G$.

Then, an induction shows that $x$ is an element of $\outset{F\plug G}$ if and only if there is a (finite, possibly empty) path from $x$ to an element of $\outset{F\cup G}$, i.e. $\outset{F\plug G}$ is at most a countable union of images of the set $\outset{F\cup G}$. But since all maps considered are non-singular, these images of $\outset{F\cup G}$ are negligible subsets since $\outset{F\cup G}$ is itself negligible. This ends the proof as a countable union of copies of negligible sets are negligible (by countable additivity), hence $\outset{F\plug G}$ is negligible.
\end{proof}

\begin{theorem}[Probabilistic model]
Let $\Omega$ be a monoid and $\microcosm{m}$ a microcosm. The set of $\Omega$-weighted \emph{sub-probabilistic} graphings in $\microcosm{m}$ yields a model, which we will denote $\pmodel{\Omega}{m}$, of multiplicative-additive linear logic.
\end{theorem}

We will now show how deterministic and probabilistic complexity classes can be characterised by means of the type of predicates $\pred{m}$ in the deterministic and probabilistic models respectively. We will start by establishing soundness by showing how the computation by automata can be represented by the execution between graphings and word representations.

\section{Soundness}

%
%

The proof of the characterisation theorem \cite{seiller-goinda} relies on a representation of multihead automata as graphings. We here generalise the result to probabilistic automata with a pushdown stack. For practical purposes, we consider a variant of the classical notion of probabilistic two-way multihead finite automata with a pushdown stack obtained by:
\begin{itemize}[nolistsep,noitemsep]
\item fixing the right and left end-markers as both being equal to the fixed symbol $\star$;
\item fixing once and for all unique initial, accept and reject states;
\item choosing that each transition step moves exactly one of the multiple heads of the automaton;
\item imposing that all heads are repositioned on the left end-marker and the stack is emptied before accepting/rejecting.
\item symbols from the stack are read by performing a $\pop$ instruction; if the end-of-stack symbol $\star$ is popped, it is pushed on the stack in the next transition.
\end{itemize}
It should be clear that these choices in design have no effect on the sets of languages recognised.

\begin{definition}\label{def:cc}
A \emph{$k$-heads probabilistic two-way multihead finite automata with a pushdown stack} (\cctwpfa{k}) $\automaton{M}$  is defined as a tuple $(\Sigma,Q,\rightarrow)$, where the transition function $\rightarrow$ is a map that associates to each element of $\starred{\Sigma}^{k}\times Q$ a sub-probability distribution over the set $\left(\textrm{Inst}\times Q\right)$ where $\textrm{Inst}$ is the set of instructions: $(\{1,\dots,k\}\times\{\In,\Out\})\times\{\identity,\pop,\pushone,\pushzero,\pushstar\}$.

The set of deterministic (resp. probabilistic) two-way multihead automata with $k$ heads is written $\twdfa{k}$ (resp. $\twpdfa{k}$) and the corresponding complexity class is noted \cctwdfa{k} (resp. \cctwpfa{k}). The set of all deterministic (resp. probabilistic) two-way multihead automata $\cup_{k\geqslant 1}\twdfa{k}$ is denoted by $\twdfas$ (resp $\twpfas$): the corresponding complexity classes \cctwdfa{$\infty$} and \cctwpfa{$\infty$} are known to be equal to $\Logspace$ and $\PLogspace$ \cite{Holzer}.

The set of $k$ heads deterministic (resp. probabilistic) two-way multihead automata with a pushdown stack is written $\twdfastack{k}$ (resp. $\twpdfastack{k}$) and the corresponding complexity class is noted \cctwdfastack{k} (resp. \cctwpfastack{k}). The set of all deterministic (resp. probabilistic) two-way multihead automata with a pushdown stack $\cup_{k\geqslant 1}\twdfastack{k}$ is denoted by $\twdfastacks$ (resp $\twpfastacks$): the corresponding complexity classes \cctwdfastack{$\infty$} and \cctwpfastack{$\infty$} are known to be equal to $\Ptime$ \cite{Macarie} and $\PPtime$.
\end{definition}

We now describe how to extend the author's translation of multihead automata as graphings to the set of all \cctwpfa{k}.
To simplify the definition, we define for all for $\automaton{t}=((\arrowedvec{s},q),(i,d',q'))$ the notation $\automaton{t}\in\rightarrow$ to denote that $\rightarrow(\arrowedvec{s},q)(i,d',q')>0$, i.e. the probability that the automaton will perform the transition $\automaton{t}$ is non-zero.

The encoding is heavy but the principle is easy to grasp. We use the stateset to keep track of the last values read by the heads, as well as the last popped symbol from the stack. The subtlety is that we also keep track of the permutation of the heads of the machine. Indeed, the graphing representation has the peculiarity that moving one head requires to use a permutation. As a consequence, to keep track of where the heads are at a given point, we store and update a permutation. Lastly, the stack is initiated with the symbol $\star$; this is done by simply restricting the source of the edges from the initial state to the subspace $V(\star)$ of sequences starting with the symbol $\star$.

\begin{definition}
Let $\automaton{M}=(\Sigma,Q,\rightarrow)$ be a \cctwpfa{k}. We define $\autograph{M}$ a graphing in $\microcosm{n}$ with dialect -- set of states -- $Q\times\mathfrak{G}_{k}\times\{\star,0,1\}^{k}\times\{\star,0,1\}$ as follows. 
\begin{itemize}
\item each transition of the form $\automaton{t}=((\arrowedvec{s},q),(\nu,q'))$ with $q\neq\init$ and $\nu=(i,d')\times\iota$ with $\iota\neq\pop$ gives rise to a family of edges indexed by a permutation $\sigma$ and an element $u$ of $\{\star,0,1\}$: 
\begin{eqnarray*}
\lefteqn{\support{(a,d)}\times\{(q,\sigma,\arrowedvec{s},u)\}}\\
&\longrightarrow& \support{(s_{i},d')}\times\{(q',\tau_{1,\sigma(i)}\circ\sigma,\arrowedvec{s}[s_{\sigma^{-1}(1)}:=s],u)\},
\end{eqnarray*} 
realised by the map $\ttrm{p}_{(1,\sigma(i))}$ together with the adequate map on the stack subspace and the adequate translation on $\integerN$, and of weight $\rightarrow(\arrowedvec{s},q)(\nu,q')$;
\item each transition of the form $\automaton{t}=((\arrowedvec{s},q),(\nu,q'))$ with $q\neq\init$ and $\nu=(i,d')\times\pop$  gives rise to a family of edges indexed by a permutation $\sigma$ and an element $u$ of $\{\star,0,1\}$: 
\begin{eqnarray*}
\lefteqn{\support{(a,d)}^{V(u)}\times\{(q,\sigma,\arrowedvec{s})\}}\\
&\longrightarrow& \support{(s_{i},d')}\times\{(q',\tau_{1,\sigma(i)}\circ\sigma,\arrowedvec{s}[s_{\sigma^{-1}(1)}:=s])\},u
\end{eqnarray*} 
realised by the map $\ttrm{p}_{(1,\sigma(i))}$ composed with the $\pop$ map and the adequate translation on $\integerN$, and of weight $\rightarrow(\arrowedvec{s},q)(\nu,q')$;
\item each transition of the form $\automaton{t}=((\arrowedvec{s},q),(\nu,q'))$ with $q=\init$ and $\nu=(i,d')\times\iota$ with $\iota\neq\pop$ gives rise to a family of edges indexed by an element $v\in\{\ttrm{a},\ttrm{r}\}$ and an element $u$ of $\{\star,0,1\}$: 
\begin{eqnarray*}
\lefteqn{\support{\ttrm{v}}^{V(\star)}\times\{(\textrm{init},\textrm{Id},\arrowedvecstar,u)\}}\\
&\longrightarrow& \support{(s_{i},d')}\times\{(q',\tau_{1,\sigma(i)},\arrowedvec{s}[s_{\sigma^{-1}(1)}:=s],u)\},
\end{eqnarray*}  
realised by the map $\ttrm{p}_{(1,\sigma(i))}$ together with the adequate map on the stack subspace and the adequate translation on $\integerN$, and of weight $\rightarrow(\arrowedvec{s},q)(\nu,q')$;
\item each transition of the form $\automaton{t}=((\arrowedvec{s},q),(\nu,q'))$ with $q=\init$ and $\nu=(i,d')\times\pop$ gives rise to a family of edges indexed by an element $v\in\{\ttrm{a},\ttrm{r}\}$ and an element $u$ of $\{\star,0,1\}$: 
\begin{eqnarray*}
\lefteqn{\support{\ttrm{v}}^{V(\star)}\times\{(\textrm{init},\textrm{Id},\arrowedvecstar,u)\}}\\
&\longrightarrow& \support{(s_{i},d')}\times\{(q',\tau_{1,\sigma(i)},\arrowedvec{s}[s_{\sigma^{-1}(1)}:=s],u)\},
\end{eqnarray*} 
realised by the map $\ttrm{p}_{(1,\sigma(i))}$ together with the $\pop$ map on the stack subspace and the adequate translation on $\integerN$, and of weight $\rightarrow(\arrowedvec{s},q)(\nu,q')$.
\end{itemize}
\end{definition}

We now generalise the key lemma from the previious paper \cite{seiller-goinda}. This result will be essential for all later results stated in this paper. The proof is a simple but lengthy induction.

\begin{proposition}\label{tracespaths}\label{prop:generalised46}
Let $\automaton{M}$ be a $\twpfa{k}$. Alternating paths of odd length between $\autograph{M}$ and $\oc \graphingtw{w}$ of source $\support{\ttrm{a}}_{Y}$ with\footnote{To understand where the subset $Y$ comes from, we refer the reader to the proof of Lemma \ref{technicallemma}.} $Y=[0,\frac{1}{\lg(\word{w})}]^{k}\times[0,1]^{\naturalN}$ are in a weight-preserving bijective correspondence with the non-empty computation traces of $\automaton{M}$ given $\word{w}$ as input.
\end{proposition}

\begin{corollary}\label{pathaccept}
The automaton $\automaton{M}$ accepts $\word{w}$ with probability $p$ if and only if $p$ is equal to the sum of the weights of alternating paths between $\autograph{M}$ and $\oc \graphingtw{w}$ of source and target $\support{\ttrm{a}}$.
\end{corollary}

We now define the probabilistic tests. These will be used to characterise probabilistic classes thanks to the lemma that follows and which relates the probability that a computation accepts with the orthogonality.

\begin{definition}
For $\eta>0$, we define the test $\testprob[\epsilon]$ as the set
\[ (\log(1-\frac{1}{2}.u),\halfidentity\support{\ttrm{a}}^{V(\star^n)}_{[0,\frac{1}{n}]^{n}\times[0,1]^{\naturalN}}])\mid u\in [0,\epsilon], n\in\naturalN \}. \]
\end{definition}

\begin{lemma}
The sum of the weights of alternating paths between $\autograph{M}$ and $\oc \graphingtw{w}$ of source and target $\support{\ttrm{a}}$ is greater than $\epsilon$ if and only if $\autograph{M}\plug\oc \graphingtw{w}\poll \testprob[\epsilon]$.
\end{lemma}

\begin{proof}
Let us write the weights of alternating paths between $\autograph{M}$ and $\oc \graphingtw{w}$ of source and target $\support{\ttrm{a}}$ as $p_0,p_1,\dots, p_k$. We use here a result from the first work on Interaction Graphs \cite{seiller-goim} showing that in the probabilistic case the measurement of two graphs $\meas{G,H}$ is equal to the measurement of the graphs $\meas{\hat{G},\hat{H}}$ where $\hat{.}$ fusion the edges with same source and target into a single edge by summing the weights \cite[Proposition 16]{seiller-goim}. Therefore, $\meas[]{\autograph{M}\plug\oc \graphingtw{w}, \testprob[\eta]}$ is equal to $\eta-\log(1-m(\frac{1}{2}\cdot\mathbf{1}.(\sum p_i)))=\eta-\log(1-\frac{1}{2}(\sum p_i))$. Now, $\sum p_i>\epsilon$ if and only if $1-\frac{1}{2} (\sum p_i)<1-\frac{1}{2}\epsilon$, if and only if $-\log(1-\frac{1}{2} (\sum p_i))>-\log(1-\frac{1}{2}\epsilon)$. I.e. $\sum p_i>\epsilon$ if and only if $\log(1-\frac{1}{2}\epsilon)-\log(1-\frac{1}{2} (\sum p_i))>0$. This gives the result, i.e. $\sum p_i>\epsilon$ if and only if $\autograph{M}\plug\oc \graphingtw{w}\poll (\log(1-\frac{1}{2}.u),\halfidentity\support{\ttrm{a}}_{[0,\frac{1}{n}]^{n}\times[0,1]^{\naturalN}}])$ for all $u\in[0,\epsilon]$.
\end{proof}

\begin{definition}
We define the complexity class $\predprob{m}$ as the set
\[ \{\mathcal{L}^{\testprob[\frac{1}{2}]}(M)\mid M\text{ $\microcosm{m}$-machine in $\pmodel{[0,1]}{m}$} \}. \]
\end{definition}

Using the preceding lemma and the definitions of the complexity classes, we obtain the following theorem.

\begin{theorem}
For all $i\in\naturalN^{\ast}\cup\{\infty\}$, 
\[\text{\cctwdfa{i}}\subseteq\preddet{m_{i}}\hspace{2em} \text{\cctwdfastack{i}}\subseteq\preddet{n_{i}}\]
\[\text{\cctwpfa{i}}\subseteq\predprob{m_{i}}\hspace{2em} \text{\cctwpfastack{i}}\subseteq\predprob{n_{i}}\]
\end{theorem}

\section{Completeness}

We here generalise a technical lemma from the previous paper \cite[Lemma4.14]{seiller-goinda} to include probabilities and pushdown stacks. The principle is the following. By the author's proof, the computation of a $\microcosm{m}_i$-machine given an input $w$ can be simulated by a computation of paths between finite graphs. This can be extended with probabilistic weights in a straightforward manner. Now, the operations on stacks could be thought of as breaking this result, since stacks are arbitrarily long. However, this can be dealt with by considering weight within the monoid $\Theta$ generated by $\{0,1,\star,c\}$ and the relations $c0=c1=c\star=\epsilon$ where epsilon is the empty sequence, thus the neutral element of $\Theta$. We will thus obtain that the computation of a $\microcosm{n}_i$-machine given an input $w$ can be simulated by a computation of paths between finite graphs with weights in $\Omega\times\Theta$.

\begin{lemma}[Technical Lemma]\label{technicallemma}
Let $M$ be a $\microcosm{n}_{\infty}$-machine. The computation of $M$ with a representation $\oc W$ of a word $\word{w}$ is equivalent to the execution of a finite\footnote{Whose size depends on both the length of $\word{w}$ and the smallest $k$ such that $M$ is a $\microcosm{m}_{k}$ machine and (TODO).} $\Omega\times\Theta$-weighted graph $\bar{M}$ and the graph representation $\graphtw{w}$ of $\word{w}$.
\end{lemma}

\begin{proof}
The proof of this lemma follows the proof of the restricted case provided in earlier work \cite{seiller-goinda}. Based on the finiteness of $\microcosm{n}_{\infty}$-machines, there exists an integer $N$ such that $M$ is a $\microcosm{n}_{N}$-machine. We now pick a word $\word{w}\in\Sigma^{k}$ and $(0,W_{\word{w}})$ the project $(0,\oc \graphingtw{w})$. All maps realising edges in $M$ or in $\oc \graphingtw{w}$ are of the form $\phi\times\identity[\bigtimes_{i=N+1}^{\infty}{[0,1]}]\times \psi$ -- i.e. they are the identity on copies of $[0,1]$ indexed by natural numbers $>N$. So we can consider the underlying space to be of the form $\integerN\times[0,1]^{N}\times\{\star,0,1\}^{\naturalN}$ instead of $\measure{X}$ by just replacing realisers $\phi\times\psi$. Moreover, the maps $\phi$ here act either as permutations over copies of $[0,1]$ (realisers of edges of $M$) or as permutations over a decomposition of $[0,1]$ into $k$ intervals (realisers of $\oc \graphingtw{w}$). Consequently, all $\phi$ act as permutations over the set of $N$-cubes $\{\bigtimes_{i=1}^{N}[k_i/k,(k_{i}+1)/k]~|~0\leqslant k_i\leqslant k-1\}$, i.e. their restrictions to $N$-cubes are translations. 

Based on this, one can build two (thick\footnote{Thick graphs are graphs with dialects, where dialects act as they do in graphings, i.e.\ as control states.}) graphs $\bar{M}$ and $\bar{W}_{\word{w}}$ over the set of vertices $\ext{\Sigma}\times\{\bigtimes_{i=1}^{N}[k_i/k,(k_{i}+1)/k]~|~0\leqslant k_i\leqslant k-1\}$ as in the proof of the restricted lemma \cite{seiller-goinda}. The only difference is that we keep track of weights and encode the stack operations as elements of $\Theta$ (we use the identification: $[[\pushone]]=1$, $[[\pushzero]]=0$, $[[\pushstar]]=\star$, $[[\pop]]=c$). There is an edge in $\bar{M}$ of source $(s,(k_{i})_{i=1}^{N},d)$ to $(s',(k'_{i})_{i=1}^{N},d')$ and weight $(p,[[\psi]])$ if and only if there is an edge in $M$ of source $\bracket{s}\times\{d\}$ and target $\bracket{s'}\times\{d'\}$, of weight $p$ and whose realisation is $\phi\times\psi$ where $\phi$ sends the $N$-cube $\bigtimes_{i=1}^{N}[k_i/k,(k_{i}+1)/k]$ onto the $N$-cube $\bigtimes_{i=1}^{N}[k'_i/k,(k'_{i}+1)/k]$. There is an edge (of weight $(1,\epsilon)$) in $\bar{W}_{\word{w}}$ of source $(s,(k_{i})_{i=1}^{N},d)$ to $(s',(k'_{i})_{i=1}^{N},d')$ if and only if $d=d'$, $k_i=k'_i$ for $i\geqslant 2$ and there is an edge in $W_{\word{w}}$ of source $\bracket{s}\times[k_1/k,(k_1+1)/k]\times[0,1]^{\naturalN}$ and target $\bracket{s'}\times[k'_1/k,(k'_1+1)/k]\times[0,1]^{\naturalN}$.

Then one checks that there exists an alternating path between $M$ and $\oc \graphingtw{w}$ of weight $p$ and whose stack operation is equal to $\psi$ if and only if there exists an alternating path between $\bar{M}$ and $\graphtw{w}$ of weight $(p,[[\psi]])$.
\end{proof}

This lemma will be useful because of the following proposition \cite[Proposition 4.15]{seiller-goinda}.

\begin{proposition}\label{orthogonalityandcycles}
For any $\microcosm{n}_\infty$-machine $M$ and word representation $\oc W$, $M\plug\oc W$ is orthogonal to $\testprob[\epsilon]$ if and only if the sum of the weights in $\Omega$ of alternating paths of $\Theta$-weight $\epsilon$ between $M$ and $\oc W\otimes \identity[\support{\ttrm{a}}]$ from $\support{\ttrm{a}}$ to itself is greater than $\epsilon$.
\end{proposition}

\begin{proof}
Using the \emph{trefoil property} for graphings \cite{seiller-goig}, which in this case becomes 
\[\meas{(0,M)\plug(0,\oc W),(t,T)}=\meas{(0,M),(0,\oc W)\plug(t,T)}.\]
Since the support of $\oc W$ and the test are disjoint, we have the equality $(0,\oc W)\plug(t,T)=(0,\oc W)\otimes(t,T)$. Hence $M\plug\oc W$ is orthogonal to $\testprob[\epsilon]$ if and only if $M$ is orthogonal to $(0,\oc W)\otimes (\log(1-\frac{1}{2}.u),\halfidentity\support{\ttrm{a}}^{V(\star^n)}_{[0,\frac{1}{n}]^{n}\times[0,1]^{\naturalN}}])$ for all $u\in [0,\epsilon]$.

Thus $M\plug\oc W$ is orthogonal to $\testdetneg$ if and only if $\log(1-\frac{1}{2}.u)+\meas{M,\oc W\otimes \halfidentity\support{\ttrm{a}}^{V(\star^n)}_{[0,\frac{1}{n}]^{n}\times[0,1]^{\naturalN}}]}\neq 0,\infty$ for all $u\in[0,\epsilon]$, i.e. if and only if there are no alternating cycles between $M$ and $\oc W\otimes \halfidentity\support{\ttrm{a}}^{V(\star^n)}_{[0,\frac{1}{n}]^{n}\times[0,1]^{\naturalN}}]$ of weight $a\cdot \mathbf{1}$ with $a\leqslant \epsilon$. Moreover, since no weights in $M$ and $\oc W$ are equal to $\lambda\cdot\mathbf{1}$, such cycles need to go through $\support{\ttrm{r}}$. 
\end{proof}

We will now define an automaton that will compute the same language as a given $\microcosm{n}_{\infty}$-machine $M$. Notice one subtlety here: a $\microcosm{n}_i$ machine can use the $\pushstar$ instruction at any given moment. Thus the automata to be defined works with a ternary stack -- over the alphabet $\{\star,0,1\}$ -- and not a binary one. This is fine because from any automaton with a ternary stack one can define an automaton on a binary stack recognising the same language (very naively, using a representation of the ternary alphabet as words of length 2, this simply multiplies the number of state by a factor of 2).

Now, a key element in the result which has not yet appeared is that the shrinking of the support of the tests as $n$ grows implies that the cycles considered in the previous proof need to go through the first $N$-cube in the finite graphs $\bar{M}$ and $\graphtw{w}$. The same trick implies that the stack is emptied during the path, i.e. the weight of the path alternating between the finite graphs $\bar{M}$ and $\graphtw{w}$ is required to be equal to $(p,c^i)$. All in all, for any $\microcosm{n}_\infty$-machine $M$ and word representation $\oc W$, $M\plug\oc W$ is orthogonal to $\testprob[\epsilon]$ if and only if there exists a path of weight $(p,c^i)$ with $p>\epsilon$ from $\support{a}_{\bigtimes_{i=1}^{N}[0,1/k]}$ -- the first $N$-cube on $\ttrm{a}$ -- to itself. It is then easy to define an automata $\autograph{M}$ that computes the same language as $M$ by simply following the transitions of $\bar{M}$, and check that this automata accepts a word $w$ with probability $p$ if and only if there is a path of weight $(p,c^i)$ with $p>\epsilon$ from $\support{a}_{\bigtimes_{i=1}^{N}[0,1/k]}$ -- the first $N$-cube on $\ttrm{a}$ -- to itself. This leads to the following proposition.

\begin{proposition}
Let $G$ be a $\microcosm{n}_i$-machine. The automaton $\autograph{G}$ is such that for all word $w$ and all word representation $!W$ of $w$, the sum of the weights in $\Omega$ of alternating paths of $\Theta$-weight $\epsilon$ between $G$ and $\oc W\otimes \identity[\support{\ttrm{a}}]$ from $\support{\ttrm{a}}$ to itself is greater than $\epsilon$ if and only if $\autograph{G}$ accepts $w$ with probability greater than $\epsilon$.
\end{proposition}

Putting together this result and the main theorem of the last section, we obtain. 

\begin{theorem}
For all $i\in\naturalN^{\ast}\cup\{\infty\}$, 
\[\preddet{m_{i}}=\text{\cctwdfa{i}}\hspace{2em} \preddet{n_{i}}=\text{\cctwdfastack{i}}\]
\[\predprob{m_{i}}=\text{\cctwpfa{i}}\hspace{2em} \predprob{n_{i}}=\text{\cctwpfastack{i}}\]
\end{theorem}

\begin{corollary}
In particular,
\[ \predprob{m_{\infty}}=\PLogspace,\hspace{2em} \predprob{n_{\infty}}=\PPtime. \]
\end{corollary}

\section{Conclusion and Perspectives}

We have shown how to extend the author's method to capture numerous complexity classes between regular languages and polynomial time, showing how the method applies as well to probabilistic computation. This provides the first examples of implicit characterisations of probabilistic complexity classes. This is however related to bounded error classes, and it will be natural to try and characterise bounded-error classes. In particular, it should be possible to capture both $\BPL$ and $\BPP$ from the present work. We expect to be able to do so using the rich notion of type provided by Interaction graphs models. As an example, let us explain how the characterisations above can be expressed through types in the models.

%
We can define the language associated to a $\microcosm{m}$-machine $M$ and a test $\testfont{T}$ as a type. Indeed, we say a word $w$ is in the langage defined by $M$ if and only if $M\plug \de{w}\poll \de{t}$ for all $\de{t}\in\testfont{T}$. Using standard properties of the execution and orthogonality \cite{seiller-goig}, this can be rephrased as $M\plug \de{t}\poll \de{w}$. Thus, $M$ defines a set of projects $\{ M\plug \de{t} \mid \de{t}\in\testfont{T}\}$ which \emph{tests} natural numbers, i.e. elements of $\ListType$.

\begin{definition}
Let $M$ be a $\microcosm{m}$-machine and $\testfont{T}$ be a test. We define the type:
\[ \LangType[M]{\testfont{T}}=(\oc \ListType^{\pol}\cup\{M\plug \de{t} \mid \de{t}\in\testfont{T}\})^{\pol} \]
\end{definition}

In fact, $\LangType[M]{\testfont{T}}$ can also be defined as an intersection type, by noting $M(\testfont{T})=\{M\plug \de{t} \mid \de{t}\in\testfont{T}\}$.
\begin{lemma}
\[ \LangType[M]{\testfont{T}} = M(\testfont{T})^{\pol}\cap\oc \ListType \]
\end{lemma}

The type represents a language in the following fashion.

\begin{proposition}
Let $M$ be a $\microcosm{m}$-machine and $\testfont{T}$ be a test. 
\[ \de{w}\in \LangType[M]{\testfont{T}} \Leftrightarrow \exists w\in \mathcal{L}^{\testfont{T}}(M), \de{w}\in\repany{w} \]
\end{proposition}

Now, this is particularly interesting when one considers that the model allows for the definition of (linear) \emph{dependent types}. Indeed, if $\cond{A}(\de{u})$ is a family of types (we suppose here that $\de{u}$ ranges over the type $\cond{U}$), the types $\sum_{\de{u}: \cond{U}}\cond{A}(u)$ and $\prod_{\de{u}: \cond{U}}\cond{A}(u)$ are well defined:
\[
\begin{array}{rcl}
\sum_{\de{u}: \cond{U}}\cond{A}(u) & = & \{ \de{u}\otimes \de{a}\mid \de{u}\in\cond{U}, \de{a}\in\cond{A}(\de{u}) \}^{\pol\pol}\\
\prod_{\de{u}: \cond{U}}\cond{A}(u) & = & \{ \de{f} \mid  \forall \de{u}\in\cond{U}, \de{f\plug u}\in\cond{A}(\de{u}) \}
\end{array}
\]

In the probabilistic model, we can use the following type to characterise\footnote{In fact, this type is more than $\PPtime$ and the latter should be defined as a quotient to identify those $M$ such that $\LangType[M]{\testfont{T}}$..} $\PPtime$:
\[\sum_{M: \oc\ListType\multimap\NBool} \LangType[M]{\testprob[\frac{1}{2}]}.\]
Indeed, we have that:
\[\cond{A}\in\PPtime \Leftrightarrow \exists M: \oc\ListType\multimap\NBool, \cond{A}=\LangType[M]{\testprob[\frac{1}{2}]}. \]

Noting that $\testprob[\frac{1}{2}]$ can be defined as a countable intersection (thus a universal quantification), it is equal to $\forall n\in\naturalN, \testprob[\frac{1}{2}+\frac{1}{n}]$. The above type then becomes:
\[\sum_{M: \oc\ListType\multimap\NBool}  \forall n\in\naturalN,~ M(\testprob[\frac{1}{2}+\frac{1}{n}])^{\pol}\cap\oc \ListType.\]
This type could then be used, through a quotient, to represent $\PPtime$ in the model $\pmodel{[0,1]}{n_\infty}$, and $\PLogspace$ in the model $\pmodel{[0,1]}{m_\infty}$.

We expect to provide types characterising bounded error predicates in the same way, providing characterisations of $\BPP$ in the model $\pmodel{[0,1]}{n_\infty}$, and $\BPL$ in the model $\pmodel{[0,1]}{m_\infty}$.
%
%

%
%

\bibliographystyle{abbrv}
\bibliography{thomas}

\begin{thebibliography}{10}

\bibitem{algebrisation}
S.~Aaronson and A.~Wigderson.
\newblock Algebrization: A new barrier in complexity theory.
\newblock {\em ACM Trans. Comput. Theory}, 1(1), 2009.

\bibitem{adams}
S.~Adams.
\newblock Trees and amenable equivalence relations.
\newblock {\em Ergodic Theory and Dynamical Systems}, 10:1--14, 1990.

\bibitem{seiller-conl}
C.~Aubert and T.~Seiller.
\newblock Characterizing co-nl by a group action.
\newblock {\em Mathematical Structures in Computer Science}, 26:606--638, 2016.

\bibitem{seiller-lsp}
C.~Aubert and T.~Seiller.
\newblock Logarithmic space and permutations.
\newblock {\em Information and Computation}, 248:2--21, 2016.

\bibitem{relativisation}
T.~Baker, J.~Gill, and R.~Solovay.
\newblock Relativizations of the $\mathrm{P}\stackrel{?}{=}\mathrm{NP}$
  question.
\newblock {\em SIAM Journal on Computing}, 4(4):431--442, 1975.

\bibitem{cobham}
A.~Cobham.
\newblock The intrinsic computational difficulty of functions.
\newblock In {\em Proceedings of the 1964 CLMPS}, 1965.

\bibitem{cookP}
S.~A. Cook.
\newblock Characterizations of pushdown machines in terms of time-bounded
  computers.
\newblock {\em J. ACM}, 18(1):4--18, Jan. 1971.

\bibitem{danosjoinet}
V.~Danos and J.-B. Joinet.
\newblock Linear logic \& elementary time.
\newblock {\em Information and Computation}, 183(1):123--137, 2003.

\bibitem{edmonds65}
J.~Edmonds.
\newblock Paths, trees and flowers.
\newblock {\em Canad. J. Math.}, 17:449--467, 1965.

\bibitem{gaboriaul2}
D.~Gaboriau.
\newblock Invariants $\ell^{2}$ de relations d'{\'e}quivalence et de groupes.
\newblock {\em Publ. Math. Inst. Hautes {\'E}tudes Sci}, 95(93-150):15--28,
  2002.

\bibitem{LLL}
J.-Y. Girard.
\newblock Light linear logic.
\newblock In {\em Selected Papers from the International Workshop on Logical
  and Computational Complexity}, LCC '94, pages 145--176, London, UK, UK, 1995.
  Springer-Verlag.

\bibitem{hartmanisstearns}
J.~Hartmanis and R.~Stearns.
\newblock On the computational complexity of algorithms.
\newblock {\em Transactions of the American Mathematical Society}, 117, 1965.

\bibitem{Holzer}
M.~Holzer, M.~Kutrib, and A.~Malcher.
\newblock Complexity of multi-head finite automata: Origins and directions.
\newblock {\em Theoretical Computer Science}, 412(1):83 -- 96, 2011.
\newblock Complexity of Simple Programs.

\bibitem{Macarie}
I.~I. Macarie.
\newblock Multihead two-way probabilistic finite automata.
\newblock {\em Theory Comput. Syst.}, 30(1):91--109, 1997.

\bibitem{MulmuleyPRAM}
K.~Mulmuley.
\newblock Lower bounds in a parallel model without bit operations.
\newblock {\em SIAM Journal of Compution}, 28(4):1460--1509, 1999.

\bibitem{GCT1}
K.~Mulmuley and M.~Sohoni.
\newblock Geometric complexity theory {I}: An approach to the {P} vs. {NP} and
  related problems.
\newblock {\em SIAM Journal of Computation}, 31(2):496--526, 2001.

\bibitem{GCT2}
K.~Mulmuley and M.~Sohoni.
\newblock Geometric complexity theory {II}: Towards explicit obstructions for
  embeddings among class varieties.
\newblock {\em SIAM J. Comput.}, 38(3), 2008.

\bibitem{GCTsurvey1}
K.~D. Mulmuley.
\newblock On p vs. np and geometric complexity theory: Dedicated to sri
  ramakrishna.
\newblock {\em J. ACM}, 58(2):5:1--5:26, Apr. 2011.

\bibitem{GCTsurvey2}
K.~D. Mulmuley.
\newblock The gct program toward the p vs. np problem.
\newblock {\em Commun. ACM}, 55(6):98--107, June 2012.

\bibitem{seiller-pramsLB}
L.~Pellissier and T.~Seiller.
\newblock Prams over integers do not compute maxflow efficiently.
\newblock submitted, 2018.

\bibitem{naturalproofs}
A.~A. Razborov and S.~Rudich.
\newblock Natural proofs.
\newblock {\em Journal of Computer and System Sciences}, 55, 1997.

\bibitem{seiller-goim}
T.~Seiller.
\newblock Interaction graphs: Multiplicatives.
\newblock {\em Annals of Pure and Applied Logic}, 163:1808--1837, December
  2012.

\bibitem{seiller-phd}
T.~Seiller.
\newblock {\em Logique dans le facteur hyperfini : g{\'e}ometrie de
  l'interaction et complexit{\'e}}.
\newblock PhD thesis, Universit{\'e} Aix-Marseille, 2012.

\bibitem{seiller-tacl}
T.~Seiller.
\newblock Measurable preorders and complexity.
\newblock Topology, Algebra and Categories in Logic Conference, 2015.

\bibitem{seiller-towards}
T.~Seiller.
\newblock Towards a \emph{Complexity-through-Realizability} theory.
\newblock http://arxiv.org/pdf/1502.01257, 2015.

\bibitem{seiller-goiadd}
T.~Seiller.
\newblock Interaction graphs: Additives.
\newblock {\em Annals of Pure and Applied Logic}, 167:95 -- 154, 2016.

\bibitem{seiller-goif}
T.~Seiller.
\newblock Interaction graphs: Full linear logic.
\newblock In {\em IEEE/ACM Logic in Computer Science (LICS)}, 2016.

\bibitem{seiller-goig}
T.~Seiller.
\newblock Interaction graphs: Graphings.
\newblock {\em Annals of Pure and Applied Logic}, 168(2):278--320, 2017.

\bibitem{seiller-masas}
T.~Seiller.
\newblock A correspondence between maximal abelian sub-algebras and linear
  logic fragments.
\newblock {\em Mathematical Structures in Computer Science}, 28(1):77--139,
  2018.

\bibitem{seiller-goinda}
T.~Seiller.
\newblock Interaction graphs: Nondeterministic automata.
\newblock {\em ACM Transaction in Computational Logic}, 19(3), 2018.

\bibitem{seiller-goie}
T.~Seiller.
\newblock {Interaction Graphs: Exponentials}.
\newblock {\em {Logical Methods in Computer Science}}, {Volume 15, Issue 3},
  Aug. 2019.

\end{thebibliography}

\end{document}